\documentclass[12pt]{article}

\usepackage{amsmath, amssymb, amstext, amsthm,epsfig}

\usepackage[usenames]{color}
\usepackage{colortbl}
\usepackage{showlabels}

\newcommand{\eps}{\varepsilon}
\newcommand{\poly}{\mathrm{poly}}
\newcommand{\N}{{\Bbb N}}
\newcommand{\cnd}{\mskip 1mu|\mskip 1mu}

\def\Bbb#1{\mathbb #1}

\renewcommand{\phi}{\varphi}
\renewcommand{\ge}{\geqslant}
\renewcommand{\le}{\leqslant}

\newtheorem{theorem}{Theorem}
\newtheorem{lemma}[theorem]{Lemma}
\newtheorem{corollary}[theorem]{Corollary}

\theoremstyle{remark}
\newtheorem{remark}{Remark}
\newtheorem{definition}{Definition}

\definecolor{nc}{rgb}{1,0.5,0}

\begin{document}

\pagestyle{plain}


\title{Some properties of antistochastic strings}
\author{Alexey Milovanov\\Moscow State University\\
{\tt almas239@gmail.com}}

\maketitle

\begin{abstract}

Algorithmic statistics is a part of algorithmic information theory (Kolmogorov complexity theory) that studies the following task: given a finite object $x$ (say, a binary string), find an `explanation' for it, i.e., a simple finite set that contains $x$ and where $x$ is a `typical element'. Both notions (`simple' and `typical') are defined in terms of Kolmogorov complexity.

It was found that this cannot be achieved for some objects: there are some ``non-stochastic'' objects that do not have good explanations. In this paper we study the properties of maximally non-stochastic objects; we call them ``antistochastic''.

It turns out the antistochastic strings have the following property (Theorem~\ref{th2}): if an antistochastic string has complexity $k$, then any $k$ bit of information about $x$ are enough to reconstruct $x$ (with logarithmic advice). In particular, if we erase all bits of this antistochastic string except for $k$, the erased bits can be restored from the remaining ones (with logarithmic advice). As a corollary we get the existence of good list-decoding codes with erasures (or other ways of deleting part of the information).

Antistochastic strings can also be used as a source of counterexamples in algorithmic information theory. We show that the symmetry of information property fails for total conditional complexity for antistochastic strings.  

\end{abstract}


\section{Introduction}

Let us recall the basic notion of algorithmic information theory and algorithmic statistics (see~\cite{shen15,lv,suv} for more details).

We consider strings over the binary alphabet $\{0,1\}$.
The set of all strings is denoted by $\{0,1\}^*$ and the length of
a string $x$ is denoted by $l(x)$. The empty string is denoted by $\Lambda$.

\subsection{Algorithmic information theory}

Let $D$ be a partial computable function mapping pairs of strings to strings.
\emph{Conditional Kolmogorov complexity} with respect to
$D$ is defined as
$$
C_D(x \cnd y)=\min\{l(p)\mid D(p,y)=x\}.
$$
In this context the function $D$ is called a \emph{description mode}
or a \emph{decompressor}. If $D(p,y)=x$
then $p$ is called a \emph{description of
$x$ conditional to $y$} or a \emph{program mapping $y$ to $x$}.

A decompressor $D$ is called \emph{universal}
if for every other decompressor $D'$ there is a string
$c$ such that $D'(p,y)=D(cp,y)$ for all $p,y$.
By Solomonoff---Kolmogorov theorem universal decompressors exist. We
pick arbitrary universal decompressor $D$ and call $C_D(x \cnd y)$ \emph{the Kolmogorov
complexity} of $x$ conditional to $y$, and denote it by $C(x \cnd y)$.
Then we define the unconditional Kolmogorov
complexity $C(x)$ of $x$ as $C(x \cnd \Lambda)$. (This version of Kolmogorov complexity is called \emph{plain} complexity; there are other versions, e.g., prefix complexity, monotone complexity etc., but for our purposes plain complexity is enough, since all our considerations have logarithmic precision.)

Kolmogorov complexity can be naturally extended to other finite objects (pairs of strings, finite sets of strings, etc.). We fix some computable bijection (``encoding'') between these objects are binary strings and define the complexity of an object as the complexity of the corresponding binary string. It is easy to see that this definition is invariant (change of the encoding changes the complexity only by $O(1)$ additive term).

In particular, we fix some computable bijection between strings and finite  subsets of $\{0,1\}^*$; the string that corresponds to a finite $A\subset  \{0,1\}^*$ is denoted by $[A]$. Then we understand $C(A)$ as $C([A])$. Similarly,  $C(x \cnd A)$ and $C(A \cnd x)$ are understood as $C(x \cnd [A])$ and $C([A]  \cnd x)$, etc.

\subsection{Algorithmic statistics}

Algorithmic statistics studies  explanations of observed data that are good in the algorithmic sense: an explanation should be simple and capture all the algorithmically discoverable regularities in the data. The data is encoded, say, by a binary string $x$. In this paper we consider explanations (statistical hypotheses) of the form ``$x$ was drawn at random from a finite set $A$ with uniform distribution''. (As argued in~\cite{vv}, the class of general probability distributions reduces to the class of uniform distributions over finite sets.)

Kolmogorov suggested in 1974~\cite{kolm} to  measure the quality of an explanation $A\ni x$ by two parameters, Kolmogorov complexity $C(A)$ of $A$ (the explanation should be simple) and the cardinality $|A|$ of $A$ (the smaller $|A|$ is, the more ``exact'' the explanation is). Both parameters cannot be very small simultaneously unless the string $x$ has very small Kolmogorov complexity. Indeed,
$C(A)+\log_2|A|\ge C(x)$ with logarithmic precision\footnote{In this paper we consider all the equations and inequalities for Kolmogorov complexities up to additive logarithmic terms ($O(\log n)$ for strings of length at most $n$).}, since $x$
can be specified by $A$ and its index (ordinal number) in $A$.
Kolmogorov called an explanation $A\ni x$ \emph{good} if
$C(A)\approx 0$ and $\log_2|A|\approx C(x)$,
that is, $\log_2|A|$
is as small as the inequality $C(A)+\log_2|A|\ge C(x)$ permits
given that  $C(A)\approx 0$. He
called a string \emph{stochastic} if it has such an explanation.

Every string $x$ of length $n$ has two trivial explanations:
$A_1=\{x\}$ and $A_2=\{0,1\}^n$. The first explanation is good when
the complexity of $x$ is small. The second one
is good when the string $x$ is random, that is, its complexity $C(x)$
is close to $n$. Otherwise, when $C(x)$ is far both from $0$ and $n$, neither of them is good.

Informally, non-stochastic strings are those having no good
explanation. They were studied in~\cite{gtv,vv}.
To define non-stochasticity rigorously we have
to introduce the notion of the \emph{profile} of $x$, which represents the parameters of possible explanations for $x$.

\begin{definition}
The  \emph{profile} of a string $x$ is the set $P_x$
consisting of all pairs $(m, l)$ of natural numbers such that there exists a finite set $A \ni x$ with $C(A) \le m$ and $\log_2|A| \le l$.
\end{definition}

Figure~\ref{f1} shows how the profile of a string $x$ of length $n$ and complexity $k$ may look like.

\begin{figure}[h]

\begin{center}
\includegraphics[scale=1]{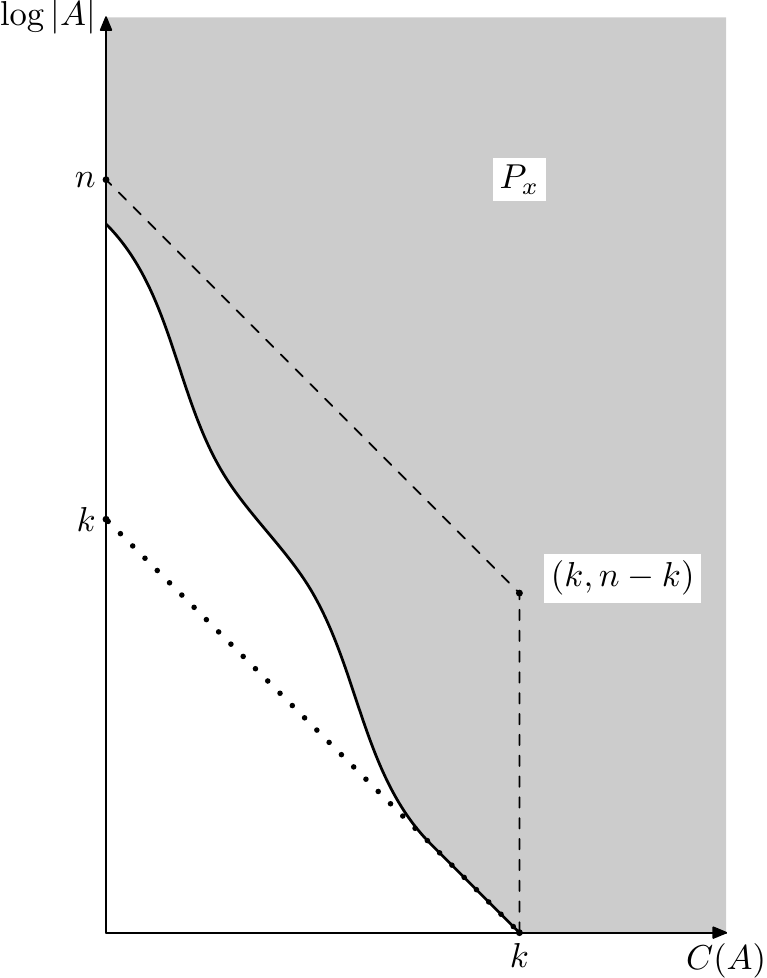}
\end{center}
\caption{The profile $P_x$ of a string $x$ of length $n$ and
complexity $k$.}
\label{f1}
\end{figure}

The profile of every string $x$ of length $n$ and complexity $k$
has the following three properties.
\begin{itemize}
\item
First,  $P_x$ is upward closed: if $P_x$ contains a pair $(m,l)$,
then $P_x$ contains all the pairs $(m',l')$ with $m'\ge m$ and $l'\ge l$.
\item
Second, $P_x$ contains the set
$$
P_{\mathrm{min}}=\{(m,l)\mid m+l\ge n \text{ or } m\ge k\}
$$
(the set  consisting of all pairs above and to the right of the dashed line
on Fig.~\ref{f1}) and is included into the set
$$
P_{\mathrm{max}}=\{(m,l)\mid m+l\ge k\}
$$
(the set
consisting of all pairs above and to the right of the dotted  line on
Fig.~\ref{f1}). In other words, the border line of $P_x$ (Kolmogorov called it the \emph{structure function} of $x$), lies between the dotted line and the dashed line.

Both inclusions are understood with logarithmic precision: the set $P_{\mathrm{min}}$ is included in the $O(\log n)$-neighborhood of the set $P_x$, and $P_x$ is included in the $O(\log n)$-neighborhood of the set $P_{\mathrm{max}}$.
\item
Finally, $P_x$ has the following property:
\begin{equation*}
\begin{split}
\text{if a pair $(m,l)$ is in $P_x$,
then}\hspace*{11em}\\
\text{the pair $(m+i+O(\log n),l-i)$ is in  $P_x$ for all $i\le l$.}
\end{split}
\end{equation*}
\end{itemize}

If for some strings $x$ and $y$ the inclusion $P_x \subset P_y$ holds, then we can say informally that $y$ is ``more stochastic'' then $x$. The largest possible profile is close to the set $P_{\mathrm{max}}$. Such a profile is possessed, for instance, by a random string of length $k$ with $n-k$ trailing zeros. As we will see soon, the minimal possible profile is close to $P_{\mathrm{max}}$; this happens for antistochastic strings.

It was shown is~\cite{vv} that every profile that has these three properties is possible for a string of length $n$ and complexity $k$ with logarithmic precision:
\begin{theorem}[Vereshchagin, Vitanyi]
\label{th1}
Assume that we are  given an upward closed
set $P$ of pairs of natural numbers
which includes $P_{\mathrm{min}}$ and is included into
$P_{\mathrm{max}}$ and for
all $(m,l)\in P$ and all $i\le l$ we have
$(m+i,l-i)\in P$.
Then there is a string $x$
of length $n$ and complexity $k+O(\log n)$ whose profile is at most
$C(P)+O(\log n)$-close to $P$.
\end{theorem}

In this theorem, we say that two subsets of
$\N^2$ are \emph{$\eps$-close} if each of them is contained in
the $\eps$-neighborhood of the other. This result mentions the complexity of the set $P$ that is not a finite set. Nevertheless, a set $P$ that satisfies the assumption is determined by the function $h(l)=\min\{m\mid (m,l)\in P\}$. This function has only finitely many non-zero values, as $h(k)=h(k+1)=\ldots=0$. Hence $h$ is a finite object, so we define the complexity of $C(P)$ as the complexity of $h$ (a finite object).

For the set  $P_{\text{min}}$ the corresponding function $h$ is defined as follows: $h(m)=n-m$ for $m<k$ and $h(k)=h(k+1)=\ldots=0$. Thus the Kolmogorov complexity of this set is $O(\log n)$. Theorem~\ref{th1} guarantees then that there is a string $x$  of length about $n$ and complexity about $k$
whose profile $P_x$ is close to the set $P_{\text{min}}$.
We call such strings \emph{antistochastic}.

The main result of our paper (Theorem~\ref{th2}) says that an antistochastic string $x$ of length $n$ can be reconstructed with logarithmic advice from every finite set $A$ that contains $x$ and has size $2^{n-k}$ (thus providing $k$ bits of information about $x$). We prove this in Section~\ref{sec:main}.

Then in Section~\ref{sec:erasures} we show that a known result about list decoding for erasure codes is a simple corollary of the properties of antistochastic strings, as well as some its generalizations. 

In Section~\ref{sec:total} we use antistochastic strings to construct an example where the so-called total conditional complexity is maximally far from standard conditional complexity: a tuple of strings $x_i$ such that conditional complexity $C(x_i \cnd x_j)$ is small while the total conditional complexity of $x_i$ given all other $x_j$ as a condition, is maximal (Theorem~\ref{thltc}).

\section{Antistochastic strings}\label{sec:main}

\begin{definition}
A string $x$ of length $n$ and complexity
$k$ is called \emph{$\eps$-antistochastic}
if for all $(m,l)\in P_x$ either
$m>k-\eps$, or $m+l>n-\eps$ (in other words, if $P_x$ is close enough to $P_{\mathrm{min}}$, see Figure~\ref{f2}).
\end{definition}

\begin{figure}[h]

\begin{center}
\includegraphics[scale=1]{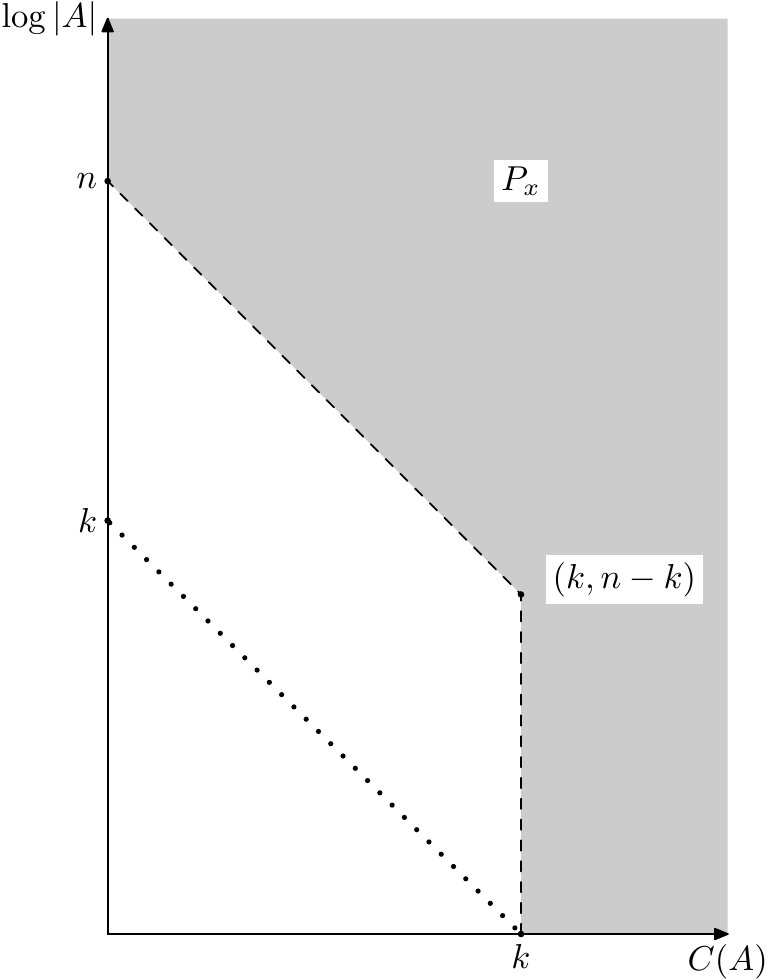}
\end{center}
\caption{The profile of an $\eps$-antistochastic string $x$ for a very small $\eps$ is close to $P_\mathrm{min}$.}
\label{f2}
\end{figure}

By  Theorem~\ref{th1} antistochastic strings exist. More precisely, Theorem~\ref{th1} has the following corollary:
\begin{corollary}\label{c1}
For all $n$ and all $k\le n$
there exists an $O(\log n)$-antistochastic string $x$ of length $n$ and complexity
$k+O(\log n)$.
\end{corollary}

This corollary can be proved more easily than the general statement of Theorem~\ref{th1}, so we reproduce its proof for the sake of completeness.

\begin{proof}
We first formulate a sufficient condition for antistochasticity.

\begin{lemma}\label{l3}
If the profile of a string $x$ of length $n$ and complexity $k$
does not contain the pair $(k-\eps,n-k)$, then $x$ is $\eps+O(\log n)$-antistochastic.
\end{lemma}

Notice that the condition of this lemma is a special case of the definition of $\eps$-antistochasticity. So Lemma~\ref{l3} can be considered as an equivalent (with logarithmic precision) definition of $\eps$-antistochasticity.

\begin{proof}[Proof of Lemma~\textup{\ref{l3}}]
Assume that a pair $(m,l)$ is in the profile of $x$. We will show that either
$m>k-\eps$ or $m+l> n-\eps-O(\log n)$. Assume that
$m\le k-\eps$ and hence $l> n-k$. By the third property of profiles we see that  the pair
$$
(m+(l-(n-k))+O(\log n), n-k)
$$
is in its profile as well. Hence we have
$$
m+l-(n-k)+O(\log n)>k-\eps
$$
and
$$
m+l>n-\eps-O(\log n).
$$
\end{proof}

We return now to the proof of Corollary~\ref{c1}.
Consider the family $\mathcal A$ consisting of all finite sets $A$ of complexity less than $k$ and log-cardinality
at most $n-k$.  The number of such sets is less than $2^k$ (they have descriptions shorter than $k$) and thus
the total number of strings in all these sets
is less than $2^{k}2^{n-k}=2^n$.
Hence there exists a string of length $n$ that
does not belong to any of sets from $\mathcal{A}$. Let $x$ be the lexicographically least such string.

Let us show that the complexity of $x$ is $k+O(\log n)$. It
is at least $k-O(1)$, as by construction the singleton $\{x\}$ has complexity at least $k$.
On the other hand, the complexity of $x$ is at most
$\log |\mathcal A|+O(\log n)\le k+O(\log n)$. Indeed, the family $\mathcal A$ can be found from $k,n$ and $|\mathcal A|$, as we can enumerate
$\mathcal A$ until we get $|\mathcal A|$ sets, and the complexity of $|\mathcal{A}|$ is bounded by $\log|\mathcal{A}|+O(1)$, while complexities of $k$ and $n$ are bounded by $O(\log n)$.

By construction  $x$ satisfies the condition of Lemma~\ref{l3} with  $\eps=O(\log n)$.
Hence $x$ is $O(\log n)$-antistochastic.
\end{proof}

Before proving our main result, set us recall some tools that are needed for it. For any integer $i$ let $\Omega_i$  denote
the number of strings of complexity at most $i$.
Knowing $i$ and $\Omega_i$, we can compute a string of Kolmogorov complexity more than $i$,
so $C(\Omega_i)=i+O(\log i)$ (in fact, one can show that $C(\Omega_i)=i+O(1)$, but logarithmic precision is enough for us).
If $l \le m$ then the leading  $l$ bits of $\Omega_m$
contain the same information as $\Omega_l$ (see~\cite[Theorem VIII.2]{vv} and \cite[Problem 367]{suv} for the proof):

\begin{lemma}\label{l2}
Assume that $l\le m$
and let $(\Omega_m)_{1:l}$ denote  the leading  $l$ bits of $\Omega_m$.
Then both  $C((\Omega_m)_{1:l} \cnd \Omega_l)$
and  $C(\Omega_l \cnd (\Omega_m)_{1:l})$ are of order $O(\log m)$.
\end{lemma}

Every  antistochastic string $x$ of complexity $k<l(x)-O(\log l(x))$ contains the same information as $\Omega_k$:

\begin{lemma}\label{l1}
There exists a constant $c$ such that the following
holds. Let  $x$ be an $\eps$-antistochastic string of length $n$
and complexity $k<n-\eps-c\log n$. Then both $C(\Omega_k \cnd x)$ and  $C(x \cnd \Omega_k)$
are less than $\eps + c\log n$.
\end{lemma}

Actually this lemma is true for all
strings whose profile $P_x$ does not contain the
pair $(k-\eps+O(\log k),\eps+O(\log k))$, in which form it
was essentially proven in~\cite{gtv}. The lemma goes back to L. Levin
(personal communication, see~\cite{vv} for details).

\begin{proof}[Proof of Lemma~\textup{\ref{l1}}]
Let us prove first that $C(\Omega_k \cnd x)$ is small.
Fix an algorithm that given $k$ enumerates all strings of
complexity at most $k$.  Let $N$ denote the number of strings
that appear after $x$ in the enumeration of all strings of complexity at most $k$ (if $x$ turns out to be the last string in this enumeration, then $N=0$).

Given $x$, $k$ and $N$, we can find $\Omega_k$ just by waiting until $N$ strings appear after $x$. If $N=0$, the statement $C(\Omega_k \cnd x) = O(\log k)$ is obvious, so we assume that $N>0$. Let $l=\lfloor\log N\rfloor$. We claim that
$ l\le \eps + O(\log n)$ because $x$ is $\eps$-antistochastic.
 Indeed, chop the set of all enumerated strings
into portions of size
 $2^l$. The last portion might be incomplete;
however $x$ does not fall in that portion since there are $N\ge 2^l$ elements after $x$. Every complete portion can be
described by its ordinal number and $k$. The total number of complete portions
is less than $O(2^k/2^l)$.
Thus the profile $P_x$ contains the
pair $(k-l+O(\log k),l)$.
By antistochasticity of $x$, we have $k-l+O(\log k)\ge k-\eps$ or
$k-l+O(\log k)+l\ge n-\eps$. The first inequality
implies that $l\le \eps+O(\log k)$. The second inequality
cannot happen provided the constant $c$ is large enough.

We see that to get $\Omega_k$ from $x$ we need only $\eps+O(\log n)$ bits of information since $N$ can be specified by $\log N=l$ bits, and $k$ can be specified by $O(\log k)$ bits.

We have shown that $C(\Omega_k \cnd x)<\eps + O(\log n)$. It remains to use the Kol\-mo\-go\-rov--Le\-vin symmetry of information theorem that says that $C(u)-C(u \cnd v)=C(v)-C(v \cnd u)+O(\log C(u,v))$ (see, e.g.,~\cite{shen15,lv,suv}). Indeed,
$$
C(x)+C(\Omega_k \cnd x)=C(x \cnd \Omega_k)+C(\Omega_k)+O(\log k).
$$
The strings $x$ and $\Omega_k$ have the same complexity with logarithmic precision, so
$C(\Omega_k \cnd x)=C(x \cnd \Omega_k)+O(\log n)$.
\end{proof}

\begin{remark}
From this lemma it follows that there are at most $2^{\eps + O(\log n)}$ $\eps$-antistochastic strings of complexity $k$ and length $n$. Indeed, we have $C(x \cnd \Omega_k) \le \eps + O(\log n)$ for each string $x$ of this type.
\end{remark}

Before stating the general result (Theorem~\ref{th2} below), let us consider its special case as example. Let us prove that every $O(\log n)$-antistochastic string $x$ of length $n$ and complexity $k$ can be restored from its first $k$ bits using $O(\log n)$ advice bits. Indeed, let $A$ consist of all strings of the same length as $x$ and having the same $k$
first bits as $x$. The complexity of $A$ is at most $k+O(\log n)$. On the other hand, the profile of $x$ contains the pair $(C(A), n-k)$. Since $x$ is $O(\log n)$-antistochastic, we have $C(A)\ge k-O(\log n)$. Therefore, $C(A)=k+O(\log n)$. Since $C(A \cnd x)=O(\log n)$, by symmetry of information we have $C(x \cnd A)=O(\log n)$ as well.

The same arguments work for every simple $k$-element subset of indices (instead of first $k$ bits): if $I$ is a $k$-element subset of  $\{1,\dots,n\}$  and $C(I)=O(\log n)$, then $x$ can be restored from   $x_I$ and some auxiliary logarithmic amount of information. Here $x_I$ denotes  the string obtained from $x$ by replacing all the symbols with indices outside $I$
by the blank symbol (a  fixed symbol different from $0$ and $1$); note that $x_I$ contains information both about $I$ and the bits of $x$ in $I$-positions.

Surprisingly, the same result is true for \emph{every}   $k$-element subset of indices, even if that subset
is complex:  $C(x \cnd x_I)=O(\log n)$.
The following theorem provides an even more general statement.

\begin{theorem}\label{th2}
Let $x$ be an $\eps$-antistochastic string of length $n$ and complexity $k$.
Assume that  a finite set $A$ is given such that $x \in A$ and $|A| \le 2^{n-k}$.
Then $C(x\cnd A) \le 2\eps + O(\log C(A) +\log n ) $.
\end{theorem}

Informally, this theorem says that any $k$ bits of information about $x$ that restrict $x$ to some subset of size $2^{n-k}$, are enough to reconstruct $x$. The $O()$-term in the right hand side depends on $C(A)$ that can be very large, but the dependence is logarithmic.

For instance, let $I$ be a $k$-element set of indices and let $A$
be the set of all strings of length $n$ that coincide with $x$
on $I$. Then the complexity of $A$ is $O(n)$ and hence
$C(x \cnd A) \le 2\eps + O(\log n ) $.

\begin{proof}
We may assume that $k < n - \eps - c \log n$ where $c$
is the constant from Lemma~\ref{l1}. Indeed, otherwise $A$ is so small ($n-k\le \eps + c$) that
$x$ can be identified by its index in $A$ in  $\eps+c$ bits.
Then by Lemma~\ref{l1}  both $C(\Omega_k \cnd x)$ and  $C(x \cnd \Omega_k)$
are less than $\eps + O(\log n)$.

In all the inequalities below we ignore additive terms of order $O(\log C(A)+\log n)$.
However, we will not ignore additive terms  $\eps$ (we do not require $\eps$ to be small, though it is the most interesting case).

Let us give a proof sketch first. There are two cases that are considered separately in the proof: $A$ is either ``non-stochastic'' or ``stochastic'' --- more precisely, appears late or early in the enumeration of all sets of complexity at most $C(A)$. The first case is easy: if $A$ is non-stochastic, then $A$ is informationally close to $\Omega_{C(A)}$ that determines $\Omega_k$ that determines $x$ (up to a small amount of auxiliary information, see the details below).

In the second case $A$ is contained in some simple small family $\mathcal{A}$ of sets; then we consider the set of all $y$ that are covered by many elements of $\mathcal{A}$ as an explanation for $x$, and use the assumption that $x$ is antistochastic to get the bound for the parameters of this explanation. This is main (and less intuitive) part of the argument.

Now let us provide the details for both parts. Run the algorithm that enumerates all finite sets of complexity at most $C(A)$, and consider $\Omega_{C(A)}$ as the number of sets in this enumeration. Let $N$ be the index of $A$ in this enumeration (so $N\le \Omega_{C(A)}$). Let $m$ be the number of common leading bits in the binary notations of $N$ and $\Omega_{C(A)}$ and let $l$ be the number of remaining bits. That is, $N=a2^l+b$ and $\Omega_{C(A)}=a2^l+c$ for some integer $a<2^m$ and $b\le c<2^l$.  For $l>0$ we can estimate $b$ and $c$ better: $ b < 2^{l-1}\le c < 2^l$. Note that $l+m$ is equal to the length of the binary  notation of $\Omega_{C(A)}$, that is, $C(A)+O(1)$. Now let us distinguish two cases mentioned:

\textbf{Case 1}: $m\ge k$. In this case we use the inequality
$C(x \cnd \Omega_k)\le \eps$. (Note that we omit terms of order $O(\log C(A)+\log n)$ here and in the following considerations.) The number   $\Omega_k$ can be retrieved from $\Omega_{m}$ since $m\ge k$ (Lemma~\ref{l2}), and the latter can be found given $m$ leading bits of $\Omega_{C(A)}$. Finally, $m$ leading bits of $\Omega_{C(A)}$ can be found given $A$, as $m$ leading bits of the index $N$ of the code of $A$ in the enumeration of all strings of complexity at most $C(A)$.

\textbf{Case 2}: $m<k$.
This case is more elaborated and we need an additional construction.

\begin{lemma}\label{lem:model}
The pair $(m,l+n-k-C(A\cnd x)+\eps)$ belongs to $P_x$.
\end{lemma}

As usual, we omit $O(\log C(A)+\log n)$ terms that should be added to both components of the pair this is statement.

\begin{proof}[Proof of Lemma~\ref{lem:model}]
We construct a set $B \ni x$ of complexity $m$ and log-size $ l+n-k-C(A\cnd x)+\eps$ in two steps.

\emph{First step}. We construct a family $\mathcal A$ of sets that is an explanation for $A$ such that $A\in \mathcal A$ and $C(\mathcal A)\le m$,  $C(\mathcal A \cnd x)\le \eps$ and $|\mathcal A|\le2^l$. To this end chop all strings of complexity at most $C(A)$ in chunks of size $2^{l-1}$ (or $1$ if $l=0$) in the order they are enumerated. The last chunk may be incomplete, however, in this case $A$ belongs to the previous (complete) chunk due to the choice of $m$ as the length of common prefix of $\Omega_{C(A)}$ and $N$.

Let $\mathcal A$ be the family of those finite sets that belong to the chunk containing $A$ and have cardinality at most $2^{n-k}$.
By construction $|\mathcal A|\le2^l$. Since $\mathcal A$ can be found from $a$ (common leading bits in $N$ and $\Omega_{C(A)}$), we have $C(\mathcal A)\le m$. To prove that $C(\mathcal A \cnd x)\le \eps$ it suffices to show that $C(a \cnd x)\le \eps$. We have
$C(\Omega_{k} \cnd x) \le\eps$ and from $\Omega_k$ we can  find $\Omega_m$ and hence the number $a$ as
the $m$ leading bits of  $\Omega_{C(A)}$ (Lemma~\ref{l2}).

\emph{Second step}. We claim that $x$ appears in at least $2^{C(A \cnd x) - \eps}$ sets from $\mathcal A$. Indeed, assume that
$x$  falls in $K$ of them. Given $x$, we need $C(\mathcal{A}\cnd  x)\le \eps $ bits to describe $\mathcal{A}$ plus $\log K$ bits to describe $A$ by its ordinal number in the list of elements of $\mathcal{A}$ containing $x$. Therefore, $C(A \cnd x)\le \log K+\eps$.

Let $B$ be the set of all strings that appear in at least  $2^{C(A \cnd x)-\eps}$ of sets
from $\mathcal A$.
As shown, $x$ belongs to $B$. As $B$ can be found from $\mathcal A$, we have $C(B)\le m$. To finish the proof of Lemma~\ref{lem:model}, it remains to estimate the cardinality of $B$. The total number of strings in all sets from
 $\mathcal A$ is at most
$2^{l}\cdot 2^{n - k}$,  and each element of $B$ is covered at least $2^{C(A\cnd x)-\eps}$ times, so $B$ contains at most $2^{l + n - k-C(A \cnd x)+\eps}$ strings.
\end{proof}

Since $x$ is $\eps$-antistochastic, Lemma~\ref{lem:model} implies that either $m\ge k-\eps$ or $m+l+n-k-C(A|x)+\eps\ge n-\eps$.
In the case $m\ge k-\eps$ we can just repeat the arguments from Case 1
and show that $C(x \cnd A)\le 2\eps$.

In the case
$m+l+n-k-C(A|x)+\eps\ge n-\eps$ we recall that $m+l=C(A)$ and by symmetry of information
$C(A)-C(A \cnd x)=C(x)-C(x \cnd A)=k-C(x 
\cnd A)$. Thus we have
$
n-C(x \cnd A)+\eps\ge n-\eps.
$
\end{proof}

\begin{remark}
Notice that every string that satisfied the claim of Theorem~\ref{th2} is $\delta$-antistochastic for $\delta\approx 2\eps$. Indeed, if $x$ has length $n$, complexity $k$ and is not $\delta$-antistochastic for some $\delta$, then $x$ belongs to some set $A$  that has $2^{n-k}$ elements
and whose complexity is less than  $k-\delta+O(\log n)$ (Lemma~\ref{l3}).
Then $C(x \cnd A)$ is large, since
$$
k=C(x)\le C(x \cnd A)+C(A)+O(\log n)\le  C(x\cnd A)+k-\delta+O(\log n)
$$
and hence $C(x \cnd A)\ge \delta- O(\log n)$ while the claim of Theorem~\ref{th2} says that $C(x\cnd A) \le 2\eps+O(\log C(A)+\log n)$.
\end{remark}

\section{Antistochastic strings and list decoding\\ from erasures}\label{sec:erasures}

Theorem~\ref{th2} implies the existence of good codes. We cannot use antistochastic strings directly as codewords, since there are only few of them. Instead, we consider a weaker property and note that every antistochastic string has this property (so it is non-empty); then we prove that there are many strings with this property and they can be used as codewords.

\begin{definition}
A string $x$ of length $n$ is called $(\eps,k)$-holographic if
for all $k$-element set of indexes $I\subset\{1,\dots,n\}$ we have $C(x \cnd x_I)<\eps$.
\end{definition}

\begin{theorem}\label{th4}
For all $n$ and all $k\le n$ there are at least $2^{k}$ strings of length $n$ that are $(O(\log n),k)$-holographic.
\end{theorem}

\begin{proof}
By Corollary~\ref{c1} and Theorem~\ref{th2} for all $n$ and $k\le n$ there exists an
$(O(\log n),k)$-holographic string $x$ of length $n$ and complexity $k$ (with $O(\log n)$ precision).  This implies that there are many of them. Indeed,
the set of all  $(O(\log n),k)$-holographic strings of length $n$ can be identified by $n$ and $k$.
More specifically, given $n$ and $k$ we can enumerate all  $(O(\log n),k)$-holographic strings
and hence $x$ can be identified by $k,n$ and its ordinal number in that enumeration.
The complexity of $x$ is at least $k-O(\log n)$, so this ordinal number is at least
$k-O(\log n)$, so there are at least $2^{k-O(\log n)}$ holographic strings. 

Our claim was a bit stronger: we promised $2^k$ holographic strings, not $2^{k-O(\log n)}$ of them. For this we can take $k'=k+O(\log n)$ and get $2^k$ strings that are $(O(\log n), k')$-holographic.  The difference between $k$ and $k'$ can then be moved into the first $O(\log n)$, since the first $k'-k$ erased bits can be provided as an advice of logarithmic size.
\end{proof}

Theorem~\ref{th4} provides a family of codes that are
list decodable from erasures. Indeed, consider $2^k$ strings that are $(O(\log n), k)$-holographic, as codewords. This code
is  list decodable from $n-k$ erasures with list size $2^{O(\log n)}=\poly(n)$. Indeed, assume that an adversary erases $n-k$ bits of a codeword $x$, so only $x_I$ remains for some set $I$ of $k$ indices. Then $x$ can be reconstructed from  
$x_I$ by a program of length $O(\log n)$. Applying all programs
of that size to $x_I$, we obtain a list of size $\poly(n)$ which contains $x$.

Although the existence of list decodable codes with such parameters
can be established by the probabilistic method~\cite[Theorem 10.9 on p. 258]{guru},
we find it interesting that a seemingly unrelated notion of antistochasticity provides such codes. In fact, a more general statement where erasures are 
replaced by any 
 other type of information loss, can be obtained in the same way.

\begin{theorem}
\label{escode}
Let $k,n$ be some integers and $k\le n$. Let $\mathcal{A}$ be a  
family of $2^{n-k}$-element subsets of $\{0,1\}^n$ that contains $|\mathcal{A}| = 2^{\poly(n)}$ subsets. Then there is a set $S$ of size at least $2^{k - O(\log n)}$ such that every $A \in \mathcal{A}$ contains at most $\poly(n)$ 
strings from 
$S$.
\end{theorem}

Theorem~\ref{th4} is a special case of this theorem: in Theorem~\ref{th4}
the family $\mathcal{A}$ consists of all sets of 
the form $\{x'\in\{0,1\}^n\mid x'_I=x_I\}$ for different $n$-bit strings 
$x$ and different sets $I$ of  $k$ indexes.

\begin{proof}
Assume first that Kolmogorov complexity of $\mathcal{A}$ 
is $O(\log n)$.
 
We use the same idea as in the proof of Theorem \ref{th4}. We may assume without loss of generality that the union of sets in $\mathcal{A}$ contains all strings, by adding some elements to $\mathcal{A}$. It can be done in such a way that 
$C(\mathcal{A})$ remains $O(\log n)$ and the size of $\mathcal{A}$
is still $2^{\poly(n)}$.

Let $x$ be an $O(\log n)$-antistochastic string of length $n$ 
and complexity~$k$. By our assumption the string $x$ belongs to some set in $\mathcal{A}$. The family $\mathcal{A}$ has low complexity 
and is not very large, 
hence for every $A \in \mathcal{A}$ we have $C(A) \le \poly(n)=2^{O(\log n)}$. 
By Theorem \ref{th2} for every $A\in\mathcal{A}$ containing $x$ we have $C(x \cnd A) < D\log n$ for some constant $D$. 

Now we define $S$ as the set of all strings $y$ such that  $C(y 
\cnd A)< D\log n$ for every $A\in \mathcal{A}$ containing $y$.
From the definition of $S$
it follows that for every $A\in \mathcal{A}$ there are at most 
$2^{D \log n}$ strings in $S$ that belong to $A$.
So now we  need to prove only that $|S| \ge 2^{k - O(\log n)}$. 

Since $C(\mathcal{A})=O(\log n)$, 
we can enumerate $S$ by a program of length $O(\log n)$. 
The antistochastic string $x$ belongs to $S$; 
on the other hand, $x$ can be identified by its ordinal number 
in that enumeration of $S$. So we conclude that the logarithm 
of this ordinal number (and therefore the log-cardinality of $S$) is at least  $k - O(\log n)$.

It remains to get rid of the assumption $C(\mathcal{A})=O(\log n)$.
To this end, fix a 
polynomial $p(n)$ in place of $\poly(n)$
in the statement of the theorem.
Then for any given $k,n$ with $k\le n$ 
consider the smallest 
$D=D_{kn}$ such that the statement of the theorem holds 
for $D\log n$ in place of $O(\log n)$. We have to show that $D_{kn}$
is bounded by a constant. For every 
$k,n$ the value $D_{kn}$  
and a family $\mathcal{A}=\mathcal{A}_{kn}$ witnessing
that $D$ cannot be made smaller that $D_{kn}$ 
can be computed by a brute force from $k,n$.
This implies that $C(\mathcal{A}_{kn})=O(\log n)$.
Hence $D_{kn}=O(1)$, as $D_{kn}$ is the worst family for $k,n$.
\end{proof}

Like Theorem~\ref{th4}, Theorem~\ref{escode}
can also be easily proved by the probabilistic method; see Theorem \ref{probproof} in Appendix.

\section{Antistochastic strings and\\ total conditional complexity}\label{sec:total}

The conditional complexity $C(a\cnd b)$ of $a$ given $b$ is defined as a minimal length of a program that maps $b$ to $a$. We may require that the program is total; in this way we get another (bigger) version of conditional complexity that was used, e.g., in~\cite{bauwens}.

Total conditional complexity is defined as the shortest length of a total program $p$ mapping $b$ to $a$: $CT(a \cnd b)=\min\{l(p)\mid D(p,b)=a$  and  $D(p,y)$ is defined for all $y\}$.

It is easy to show that the total conditional complexity may be much higher 
than the plain conditional complexity (see, e.g.,~\cite{shen12}). Namely, there exist strings $x$ and $y$ of length $n$ such that $CT(x \cnd y) \ge n $ and $C(x \cnd y) = O(1)$. Antistochastic strings help to extend this result (unfortunately, with slightly worse accuracy):
 
\begin{theorem}
\label{thltc}

For every $k$ and $n$ there exist strings $x_1 \ldots x_k$ of length $n$ such that:

\textup{(1)} $C(x_i \cnd x_j) = O(\log k+\log n)$ for every $i$ and $j$.

\textup{(2)} $CT(x_i \cnd  x_1 \ldots x_{i-1}x_{i+1} \ldots x_k) \ge n - O(\log k+\log n)$ for every $i$.
\end{theorem}  

\begin{proof}
 
Let $x$ be an $O(\log (kn))$-antistochastic string of length $kn$ and complexity $n$. We consider $x$ as the concatenation of $k$ strings of length $n$: $$x=x_1 \ldots x_k,\qquad
x_i \in \{0,1\}^n.$$ Let us show that the strings $x_1, \ldots, x_k$ satisfy the requirements of the theorem. 

The first statement is a simple corollary of antistochasticity of $x$. Theorem~\ref{th2} implies that $C(x \cnd x_j) = O(\log (kn))$ for every $j$. As $C(x_i \cnd x) = O(\log (kn)$ for every $i$, we have 
$C(x_i \cnd x_j)=O(\log (kn))$ for every $i$ and $j$.

To prove the second statement consider a total program $p$ such that $p(x_1\ldots x_{i-1} x_{i+1}\ldots x_k) = x_i$. Our aim is to show that $p$ is long.  Change $p$ to a total program  $\tilde{p}$ such that
  $\tilde{p}(x_1\ldots x_{i-1}x_{i+1}\ldots x_k) = x$ and $l(\tilde{p}) \le l(p) + O(\log (kn))$.
  Consider the set  
$$A := \{\tilde {p}(y)\mid y \in \{0, 1\}^{k(n-1)} \}.$$ 
Note that $A$ contains antistochastic string $x$ of length $kn$ and complexity~$n$ and $\log|A| \le k\cdot (n-1)$.  By the definition of antistochasticity we get $C(A) \ge n - O(\log (kn))$. By the construction  of $A$ it follows that
$$ C(A) \le l(\tilde{p})+O(\log (kn)) \le l(p)+O(\log (kn)).$$
So, we get $l(p) \ge n-O(\log (kn))$, i.e., $$CT(x_i \cnd x_1 \ldots x_{i-1}x_{i+1}\ldots x_k)\ge n-O(\log(kn)).$$
\end{proof}

\begin{remark}
This example, as well as the example from~\cite{ver},
shows that for total conditional complexity the symmetry of information does not hold.
Indeed, let $CT(a)=CT(a \cnd \Lambda)=C(a)+O(1)$.
Then 
$$
CT(x_1) - CT(x_1 
\cnd x)=(n+O(\log kn))-O(\log k)=n+O(\log kn)
$$ 
while 
$$
CT(x) - CT(x \cnd x_1)= (n+O(\log kn))- (n+O(\log kn))=O(\log kn)
$$
for strings $x,x_1$ from Theorem~\ref{thltc}.

A big question in time-bounded Kolmogorov complexity is
whether the symmetry of information 
holds for time-bounded Kolmogorov complexity.
Partial answers to this question were obtained in \cite{lm,lw,lr}.
Total conditional complexity $CT(b 
\cnd a)$ is defined as the shortest length of a total
program $p$ mapping $b$ to $a$. Being total that program halts on all inputs
in time bounded by a total computable function $f_p$ of its input. Thus
total conditional complexity may be viewed as a variant of
time bounded conditional complexity.
Let us stress that the upper bound $f_p$ for time
may depend (and does depend) on $p$  in a non-computable way.
Thus $CT(b 
\cnd a)$ is a rather far approximation to
time bounded Kolmogorov complexity.
\end{remark}

\section*{Acknowledgments}

I would like to thank Alexander Shen and Nikolay Vereshchagin for useful discussions, advice and remarks.

\section*{Appendix}

Here we provide a probabilistic proof of Theorem \ref{escode}:

\begin{theorem}
\label{probproof}
Let $k,n$ be some integers, $k\le n$. Let $\mathcal{A}$ be a finite
family of 
$2^{n-k}$-element subsets of $\{0,1\}^n$.
Then there is a set $S$ of size at least $2^k$ such that every $A \in \mathcal{A}$ contains at most $\log|\mathcal{A}|+1$ strings from $S$.
\end{theorem}

\begin{proof}
Let us show that a randomly chosen subset of $\{0,1\}^n$ of size approximately $2^k$ has the required property with a positive probability. More precisely, we assume that every $n$-bit string is included in $S$ independently with probability $\frac{1}{2^{n-k-1}}$. 

The cardinality of $S$ is the random variable with binomial distrubution. 
The expectation of $|S|$ is equal to $2^n \cdot \frac{1}{2^{n-k -1}} = 2^{k+1}$. The variance $\sigma^2$ of $|S|$ is equal to $2^n \cdot \frac{1}{2^{n-k -1}} \cdot (1 - \frac{1}{2^{n-k -1}}) \le 2^{k+1} $. By Chebyshev's inequality 
$$
P(| \cnd S| - E| \ge \frac{1}{2} \sigma) \le \left(\frac{1}{2}\right)^2 \Rightarrow P(| \cnd S| - 2^{k+1}| \ge 2^k) \le \frac{1}{4}.  
$$ 
Hence, the event  ``$|S| < 2^k$'' happens with probability at most 
$\frac{1}{4} < \frac{1}{2}$.

It remains to show that the event 
``there is a set in $ \mathcal{A}$ containing more than $2^{O(\log n)}$ 
strings from $S$'' happens with probability less than $\frac{1}{2}$.
To this end we show that for every $A \in \mathcal{A}$ the event ``$A$ contains more than $\log|\mathcal{A}|+1$ elements of $S$'' 
has probability less than $\frac{1}{2} \cdot \frac{1}{|\mathcal{A}|}$. 

Fix a $2^{n-k}$-element set $A\in\mathcal A$. 
For every $i$ the probability 
that $S$ contains at least $i$ elements from $S$ is at most
$$
\binom{|A|}{i}\cdot 2^{-(n-k +1)i} \le \frac{|A|^{i}}{i!} \cdot 2^{(-n+k-1) \cdot i} = \frac{2^{-i}}{i!}\le 2^{-i}.
$$
This value is less than $\frac{1}{2}\frac{1}{|\mathcal{A}|}$
for  $i =\log|\mathcal{A}|+1$.
\end{proof}

\end{document}